\documentclass[%
 reprint,
 amsmath,amssymb,
 aps,bm,
showpacs]{revtex4}
\usepackage{cleveref}
\usepackage[utf8]{inputenc}
\usepackage{amsmath, mathtools, tensor}
\usepackage{graphicx}
\usepackage{color}
\usepackage{dcolumn}
\usepackage{bm}
\usepackage{comment}
\usepackage{esint}
\usepackage{float}

\usepackage{graphicx} 
\usepackage{support-caption}

\usepackage[compatibility=false]{caption}

\usepackage{subcaption}

\usepackage{changepage}
\usepackage{cool}
\usepackage{xcolor}

\newcommand{\closedint}[2]{{\ensuremath{\left [ #1, #2 \right ]}}}
\newcommand{\paren}[1]{{\ensuremath{\left ( #1 \right )}}}
\newcommand{\size}[1]{{\ensuremath{\left | #1 \right |}}}

\newcommand{\map}[2]{{\ensuremath{#1 \left ( #2 \right )}}}
\newcommand{\set}[1]{{\ensuremath{\left \{ #1 \right \}}}}

\newcommand{\sgn}{\text{sgn}}

\newcommand{\dd}{\ensuremath{\textrm{d}}}

\newcommand{\UD}[2]{\ensuremath{^{#1}_{\phantom{#1} #2}}}

\newcommand{\DU}[2]{\ensuremath{_{#1}^{\phantom{#1} #2}}}

\newcommand{\UDDD}[4]{\ensuremath{^{#1}_{\phantom{#1} #2 #3 #4}}}

\newcommand{\RR}{\ensuremath{{\mathbb{R}}}}

\newcommand{\beq}{\begin{equation}}
\newcommand{\eeq}{\end{equation}}
\newcommand{\bea}{\begin{eqnarray}}
\newcommand{\eea}{\end{eqnarray}}
\newcommand{\bean}{\begin{eqnarray*}}
\newcommand{\eean}{\end{eqnarray*}}
\newcommand{\bit}{\begin{itemize}}
\newcommand{\eit}{\end{itemize}}
\newcommand{\bfi}{\begin{figure}}
\newcommand{\efi}{\end{figure}}
\newcommand{\bfic}{\begin{figure*}}
\newcommand{\efic}{\end{figure*}}
\newcommand{\bce}{\begin{center}}
\newcommand{\ece}{\end{center}}
\newcommand{\bt}{\begin{table}}
\newcommand{\et}{\end{table}}
\newcommand{\btb}{\begin{tabular}}
\newcommand{\etb}{\end{tabular}}

\newcommand{\mba}{\mathbf A}
\newcommand{\mbb}{\mathbf B}
\newcommand{\mbc}{\mathbf C}

\newcommand{\calA}{\ensuremath{\mathcal{A}}}
\newcommand{\calP}{\ensuremath{\mathcal{P}}}

\newcommand{\calE}{\ensuremath{\mathcal{E}}}
\newcommand{\calO}{\ensuremath{\mathcal{O}}}

\newcommand{\WXL}{\ensuremath{W_{XL}}}
\newcommand{\WXX}{\ensuremath{W_{XX}}}
\newcommand{\WLX}{\ensuremath{W_{LX}}}
\newcommand{\WLL}{\ensuremath{W_{LL}}}
\newcommand{\tWXL}{\ensuremath{\tilde W_{XL}}}

\newtheorem{theorem}{Theorem}[section]

\newtheorem{proposition}[theorem]{Proposition}
\newtheorem{corollary}[theorem]{Corollary}
\newtheorem{ddefinition}[theorem]{Definition}

\newenvironment{proof}[1][Proof]{\begin{trivlist}
\item[\hskip \labelsep {\bfseries #1}]}{\end{trivlist}}
\newenvironment{definition}[1][Definition]{\begin{trivlist}
\item[\hskip \labelsep {\bfseries #1}]}{\end{trivlist}}

\newcommand{\qed}{\nobreak \ifvmode \relax \else
      \ifdim\lastskip<1.5em \hskip-\lastskip
      \hskip1.5em plus0em minus0.5em \fi \nobreak
      \vrule height0.75em width0.5em depth0.25em\fi}

\begin{document}
\preprint{APS/123-QED}

\title{Testing the null energy condition with precise distance measurements}

\author{Miko\l{}aj Korzy\'n{}ski }
 \email{korzynski@cft.edu.pl}
 \author{Julius Serbenta}
 \email{julius@cft.edu.pl}
\affiliation{%
 Center for Theoretical Physics, Polish Academy of Sciences,
Al. Lotnik\'ow 32/46, 02-668 Warsaw, Poland
}%

\begin{abstract}
We present an inequality between two types of distance measures from an observer to a single light source in general relativity. It states that for a given source and observer the distance  measured by the trigonometric parallax is never shorter than the angular diameter distance provided that the null energy condition holds and that there are no focal points in between. This result is independent of the details of the spacetime geometry or the motions of the observer and the source. The proof is based on the geodesic bilocal operator formalism and on well-known properties of infinitesimal light ray bundles. Observation of the violation of the distance inequality would mean that on large scales either the null energy condition does not hold or that light does not travel along null geodesics.

\end{abstract}


\maketitle

\section{Introduction}

Measuring the distance to a given light source is one of the fundamental problems of astronomy. Many methods have been developed for that purpose depending on the nature of the source and the distance.  The three probably most straightforward ones are the trigonometric parallax measurement, the determination of the angular size of an extended source of known size (i.e. a standard ruler) and the measurement of the energy flux from an isotropic source with known absolute luminosity (i.e. a standard candle). By definition, all three methods must give the same result in a flat spacetime and in the absence of relative motions between the source and the observer. However, in the presence of spacetime curvature and relative motions the three methods are  inequivalent. Therefore, in general relativity (GR) we distinguish the angular diameter distance $D_{ang}$, also known as the area distance, defined via the solid angle occupied by the image of a standard ruler, the luminosity distance $D_{lum}$, defined by the measured energy flux from a standard candle, and the parallax distance $D_{par}$, defined via the apparent displacement of the image given the displacement of the observer along a baseline \cite{perlick2004r, grasso2019}.
Note that in the presence of curvature the trigonometric parallax may depend on the baseline orientation. However, it is possible to define a baseline-averaged quantity, combining the parallax effects from two orthogonal directions \cite{grasso2019}. 

 $D_{ang}$, $D_{lum}$ and the redshift $z$ measured between  a fixed source and observer  are related by the well-known Etherington's reciprocity relation $D_{lum}= (1+z)^2\,D_{ang}$ independently of the spacetime geometry \cite{etherington1933, etherington2007, penrose1966, kristian1966, ellis1971, ellis2009rep, schneider1992, perlick2004r, uzun2020}. In case of the baseline-averaged parallax distance $D_{par}$ the relation to other distance measures is more complicated and depends on the curvature tensor along the line of sight. In fact, for short distances the relative difference between $D_{ang}$ and $D_{par}$, called the \emph{distance slip} $\mu$, depends only on the matter content along the line of sight: the leading order correction is given by an integral of the component $T_{\mu\nu}\,l^\mu\,l^\nu$ of the stress-energy tensor, where $l^\mu$ is the tangent vector to the null geodesic connecting the observer and the source. Namely, for short distances we have
\bea
\mu = 1 - \frac{D_{ang}^2}{D_{par}^2} = 8\pi G\,\int_{\lambda_\calO}^{\lambda_\calE} T_{\mu\nu}\,l^\mu\,l^\mu\,(\lambda_{\calE}-\lambda)\,\dd\lambda + O\left(\textrm{Riem}^2\right), \label{eq:relation}
\eea
where $\lambda$ is the affine parameter of the connecting null geodesic, $\lambda_\calO$ corresponds to the observation point and $\lambda_{\calE}$ to the source, while $O\left(\textrm{Riem}^2\right)$ denotes terms involving higher powers of the curvature \cite{grasso2019, korzynski2020}.

In this paper we show that the sign of the difference between $D_{par}$ and $D_{ang}$ for a given source and a given observer is directly related to the \emph{null energy condition} (NEC). On the perturbative level this already follows from \eqref{eq:relation}: the leading, linear term in curvature is obviously nonnegative if the NEC condition holds, because in this case we have $0 \le R_{\mu\nu}\,l^\mu\,l^\nu = 8\pi G\,T_{\mu\nu}\,l^\mu\,l^\nu$ for any null $l^\mu$. The main theorems of this paper extend this inequality to the  non-perturbative level:  we show that if the NEC is satisfied then $D_{par} \ge D_{ang}$  at least up to a well-defined, finite distance between the observer and the light source. More precisely, for a fixed observation point and variable source position along a null geodesic the inequality is guarateed to hold from the observation point up to the so-called focal point, where the parallax distance reaches infinite value. Note that this blow-up may  happen for a finite value of the affine parameter. Past that point, however, $D_{par}$ becomes finite again and it is possible that $D_{par} < D_{ang}$ there even if the NEC holds globally \cite{serbenta2021}. This restriction means that the distance inequality is not global. 
Stated in a more physical language, the inequality  means that both the matter along line of sight, resulting in Ricci focusing of the null geodesics, and the tidal forces, producing their shear, can only increase $D_{par}$ in comparison with $D_{ang}$, at least up to the first focal point.

 Recall that the \emph{null energy condition} (NEC) for all null vectors $l^\mu$ reads
\bea
 T_{\mu\nu}\,l^\mu\,l^\nu \ge 0.
\eea
It is one of the weakest classical pointlike energy conditions with relatively simple and important applications. In general relativity, it is equivalent to the null convergence condition (NCC)
\bea
 R_{\mu\nu} l^\mu l^\nu \ge 0,
\eea 
which ensures that at every point light rays experience gravity as an attractive force. This property implies that any light ray bundle eventually has to reach a caustic over a finite distance provided that suitable initial condition holds. More precisely, the standard focusing theorem states that the ray bundle must reach a point where its expansion diverges provided that the expansion is negative at a single point, and NCC holds \cite{senovilla1998}. Surprisingly, this simple result is necessary for proving such basic black hole properties like formation of event horizons and singularities under the gravitational collapse, some of the black hole laws, various no hair theorems, and some versions of the positivity of ADM mass \cite{curiel2017}.

As we have already noted, if we assume the Einstein equations then the NEC is directly related to the NCC. It is satisfied for most reasonable types of matter like dust, radiation, fluid, or classical electromagnetic fields \cite{maeda2020}. It is also insensitive to the presence of cosmological constant.  The validity of the NEC can also be used to put bounds on various properties of the FLRW Universe \cite{cattoen2008}. On the other hand, it has also been noted that quantum effects tend to violate the NEC as well as its even weaker averaged version \cite{barcelo2002}. 
The NEC can also fail in the presence of non-standard matter fields \cite{barcelo2002, rubakov2014, curiel2017}, including fluids with barotropic index $w < -1$ or for holographic dark energy models with even smaller barotropic index \cite{colgain2021critique}. 

In modified theories of gravity NEC and NCC are in general not equivalent. The NCC can fail even in presence of standard matter if the field equations contain additional terms, like in the case of $f(R)$ gravity \cite{santos2007, sotiriou2010} or other extended metric theories \cite{baccetti2012, capozziello2014, capozziello2015}. Moreover, light propagation may also work differently than in GR: outside the Riemannian geometry regime the light may follow null curves which are neither autoparallel nor extremal, and this implies that the optical equations contain additional geometric terms, effectively acting as additional, nonclassical matter fields \cite{bolejko2020,speziale2018,santana2017}.  These terms affect the focusing and defocusing properties of the spacetime and, consequently, the relation between NEC and the properties of light rays.

In this paper we assume that standard general relativity holds. Therefore, we will assume NEC, which we in turn treat as equivalent to NCC. However, the reasoning should still apply to any metric theory in which light travels along null geodesics and NCC holds.

 The main result of this paper, i.e. the distance inequality, provides a method to test the NEC using sufficiently precise, simultaneous distance measurements by parallax and angular (or luminosity) to a number of light sources. The observation of the violation of the NEC would require a serious reevaluation of the fundamentals of physics. We point out, however, that the precision required seems beyond what is currently possible.

The proof of the main theorem of the paper, that is, Theorem \ref{thm:mu}, makes use of the bi-local approach to light propagation in curved spacetimes, developed  in \cite{grasso2019, serbenta2021}, and of the standard infinitesimal ray bundle formalism \cite{perlick2004r}, closely related to the null congruence formalism. 
We first show that the distance slip $\mu$ between the observation point $\calO$ and any  source located along a past-directed null geodesic from $\calO$ is related to the ratio of the cross-sectional areas  of a particular infinitesimal bundle of null rays,  with cross sections taken at  $\calO$ and at the source. We then show that this cross-sectional  area cannot increase along the null geodesic 
as we move away  from the observation point  in the past direction, at least up to the first focal point, which completes the proof.  The main argument is  similar to the reasoning used in the proof of the standard focusing theorem.

\paragraph{Notation and conventions.}
Greek letters $\paren{\alpha,\beta,\dots}$ run from $0$ to $3$, lowercase Latin indices run from $1$ to $3$ and uppercase Latin indices run from $1$ to $2$. They all enumerate tensor components in the coordinate tetrad. Boldface versions of indices cover the same range but denote components in a parallel transported tetrad. The dot stands for the derivative with respect to the affine parameter along the null geodesic. Subscript $\mathcal O$ and $\mathcal E$ denote evaluation of the quantity at respectively the point of observation and emission, i.e. $f_{\mathcal {O}} \equiv \map f {\lambda_{\mathcal O}}$. We assume the speed of light $c=1$ throughout the paper.

\paragraph{Structure of the paper.} 
In Sec. II, we briefly review the bilocal geodesic operator (BGO) formalism and relate these operators to the magnification and parallax matrices as well as angular diameter and parallax distances. Then we introduce the notion of the inifinitesimal ray bundle and present two types of bundles which we will use later. Finally, we recall Sachs optical equations and their connection to the BGOs. In Sec. III we present our main results. In the first part we prove the inequality for the parallax distance $D_{par}$ with baseline averaging performed using the determinant of the parallax matrix. In the second part we show that a similar conclusion follows for the parallax distance with the baseline averaging via trace, as proposed earlier by a number of authors. Lastly, we explain why our result cannot be extended to single baseline parallax measurements. We gather our final remarks in Sec. IV, including a short discussion of the prospects for  measurement.

\section{Preliminaries}

Let $\paren {M, g}$ be a Lorentzian spacetime, with signature $(-,+,+,+)$. Let $\calO, \calE \in M$ be points contained in a geodesically convex set. $\calO$ will denote the observation point, lying in the causal future of the emission point $\calE$. Let $\gamma_0 : \closedint {\lambda_
\calO} {\lambda_\calE} \to M$ be the unique geodesic connecting $\calO$ and $\calE$, and we assume here that $\gamma_0$ is null. By convention we assume that the affine parameter $\lambda$ runs backwards in time, i.e. $\lambda_\calO < \lambda_\calE$. Let $N_\calO, N_\calE \subset M$ be locally flat neighbourhoods of $\calO$ and $\calE$ respectively extending in all four dimensions. Let $l^\mu$ be the vector tangent to $\gamma_0$. Let $\xi^\mu$ denote the deviation vector (Jacobi field) along $\gamma_0$, satisfying the first order \emph{geodesic deviation equation} (GDE):
\bea
\nabla_l \nabla_l \xi^\mu = R \UD{\mu}{ll\nu}\,\xi^\nu,
\label{eq:gde}
\eea
where the \emph{optical tidal tensor} $ R \UD{\mu}{ll\nu}$ is defined as $ R \UD{\mu}{ll\nu} \equiv R\UD{\mu}{\alpha\beta\nu}\,l^\alpha\,l^\beta$. This is a linear, second order ordinary differential equation for $\xi^\mu$. It is possible to rewrite it  as an equation for four bitensors, forming together the formal resolvent of the GDE \cite{grasso2019}. In this language the general solution to \eqref{eq:gde} can be expressed as
\bea
\begin{split}
\xi^\mu(\lambda) &= {\WXX} \UD \mu \nu(\lambda)\, \xi^\nu( \lambda_\calO) + {\WXL} \UD \mu \nu(\lambda) \,{\nabla_l \xi^\nu}(\lambda_\calO)\\
\nabla_l \xi^\mu(\lambda) &= {\WLX} \UD \mu \nu (\lambda)\,\xi^\nu(\lambda_\calO) + {\WLL} \UD \mu \nu (\lambda)\,\nabla_l \xi^\nu (\lambda_\calO),
\end{split}
\eea
where $\WXX$, $\WXL$, $\WLX$ and $\WLL$ are 4 bitensors, or two-point tensors, acting from the observation point $\calO$ to the point $\lambda$ on the geodesic, called the \emph{bilocal geodesic operators} (BGO's). The BGO's are functionals of the curvature along $\gamma_0$ defined via appropriate ordinary differential equations involving the components of the optical tidal tensor as coefficients \cite{korzynski2021, grasso2019}.

We introduce the \emph{seminull tetrad} (SNT)  of the form $(u^\mu, e_{\mba}^\mu, l^\mu)$, parallel propagated along $\gamma_0$. It comprises the null tangent vector $l^\mu$, two orthonormal spacelike vectors $e^\mu_\mba$, ${\mba} = 1,2$,  orthogonal to $l^\mu$, spanning a spacelike two-dimensional subspace called the screen space or transverse space, and a normalized timelike vector $u^\mu$ orthogonal to $e^\mu_\mba$. $u^\mu$ is commonly identified with the 4-velocity of an observer measuring the position and the direction of the propagation of photons. It satisfies $u^\mu\, l_\mu = Q$, where $Q$ is a positive constant. Given the fixed observation point $\calO$, the equations for the transverse components of $W_{XX}$ and $W_{XL}$  expressed in the SNF read \cite{grasso2019}:
\bea
\begin{split}
{{\ddot W}{}_{XX}}\UD{\mba}{\mbb} - R\UD{\mba}{ll\mbc}\,{W_{XX}}\UD{\mbc}{\mbb} &= 0 \\
{W_{XX}}\UD{\mba}{\mbb}(\lambda_{\calO}) &= \delta\UD{\mba}{\mbb} \\
{\dot W_{XX}}{}\UD \mba \mbb(\lambda_\calO) &= 0
\end{split}
\label{eq:WXXdef}
\eea
and
\bea
\begin{split}
{{\ddot W}{}_{XL}}\UD{\mba}{\mbb} - R\UD{\mba}{ll\mbc}\,{W_{XL}}\UD{\mbc}{\mbb} &= 0 \\
{W_{XL}}\UD{\mba}{\mbb}(\lambda_{\calO}) &=  0 \\
{\dot W_{XL}}{}\UD \mba \mbb(\lambda_\calO) &= \delta\UD{\mba}{\mbb}.
\end{split}
\label{eq:WXLdef}
\eea

With this machinery we may introduce the parallax matrix, the magnification matrix and the optical distance measures. Suppose we project the BGOs onto the SNT, with the timelike vector corresponding to the observer's 4-velocity $u_\calO^\mu$. Their projection onto the screen space spanned by $e^\mu_\mba$ is related to the matrices. 

The magnification matrix $M\UD{\mba}{\mbb}$ relates the transverse vectors representing the spatial extent of a luminous body to the vectors on the screen space representing the angular size on an observer's sky \cite{grasso2019} (in the gravitational lensing theory this quantity is usually defined using angular variables, and therefore rescaled with respect to $M\UD{\mba}{\mbb}$ as defined here). We have
\bea
M \UD \mba \mbb = \frac 1 {\paren{u_\calO^\mu\, l_{\calO\, \mu}} } \left({\WXL}^{-1}\right) \UD \mba \mbb, \label{eq:Mdef}
\eea
see \cite{grasso2019}. The submatrix ${\WXL}\UD\mba\mbb$, whose inverse appears in the formula above, is also known as the Jacobi matrix.

The parallax matrix on the other hand relates the transverse displacement of an observer with the apparent change of the position of a body on the observer's sky \cite{grasso2019} and is given by
\bea
\Pi \UD \mba \mbb = \frac 1 {\paren{u^\mu_\calO\, l_{\calO\,\mu}}} \left({\WXL}^{-1}\right) \UD \mba \mbc\, {\WXX} \UD \mbc \mbb.
\eea
In a flat spacetime both linear operators are proportional to a unit matrix. Therefore both the trigonometric parallax and the magnification do not depend on the direction of the baseline or the orientation of the source's shape. However, in the general case the direction in the transverse space is important. Both linear operators can be used to define direction-averaged distances to the observed source of light. In the BGO formalism it is natural to do this as follows: the angular diameter distance, or the area distance, is defined via the determinant of $M \UD \mba \mbb$:
\bea
D_{ang} &=& \size{\det M \UD \mba \mbb}^{-1/2} = \paren {u^\mu_\calO\, l_{\calO \mu}}\, \size{\det {\WXL} \UD \mba \mbb}^{1/2}\label{eq:distDang}.
\eea
In a more operational language, it is given by the ratio of the cross-sectional area of a luminous object and the solid angle taken by its image. It is a well-known quantity in relativistic geometrical optics.
The parallax distance averaged by determinant has been introduced in \cite{grasso2019}:
\bea
D_{par} &=& \size{\det \Pi \UD \mba \mbb}^{-1/2} = \paren {u^\mu_\calO\, l_{\calO \mu}} \,\size{\det {\WXL} \UD \mba \mbb}^{1/2} \,\size{\det {\WXX} \UD \mba \mbb}^{-1/2}.
\label{eq:dists}
\eea
Both distance measures depend on the state of motion of the observer, given by $u_\calO^\mu$, and the spacetime geometry along the line of sight. However, they do not depend on the state of motion of the light emitter.

Finally, we recall the definition of the main quantity of this work, i.e the \emph{distance slip} $\mu$ \cite{grasso2019}:
\bea
\mu = 1 - \sigma \frac{D_{ang}^2}{D_{par}^2} = 1 - \frac{\det \Pi\UD{\mba}{\mbb}}{\det M\UD{\mba}{\mbb}},\label{eq:muD}
\eea
where $\sigma = \pm 1$ is a sign correction given by $\sigma = \sgn \det \Pi\UD{\bm A}{\bm B} / \sgn \det M \UD{\bm A}{\bm B}$.
It is independent of the states of motion of both the observer and the source. It vanishes in a flat spacetime, but not necessary in a curved one.

From \eqref{eq:Mdef}-\eqref{eq:dists} we can rewrite the equations above as
\bea
\mu = 1 - \det {W_{XX}}\UD{\mba}{\mbb} \label{eq:muWXX}
\eea
and
\bea
 \sigma = \sgn \det {W_{XX}}\UD{\mba}{\mbb}. \label{eq:sigmaWXX}
 \eea

\subsection{Infinitesimally thin bundles}

The infinitesimally thin ray bundle formalism is a complementary method to describe light propagation between two points. We make use of it in this paper, because the propagation equations describing the bundles 
can be easily used to prove inequalities involving observable quantities. We will briefly review the infinitesimally thin ray bundle formalism, as described in \cite{perlick2004r}. 

By an infinitesimal ray bundle with an elliptical cross section we mean the set

\bea
S = \set {c^I \xi^\mba_I ~|~ c^1,c^2 \in \RR, \, c^I c^J \delta_{I J} \le 1}
\eea
where $\xi^\mba$ satisfies the GDE in the SNT:
\bea
\ddot \xi^\mba_I &= R \UDDD \mba l l \mbb \, \xi^\mbb_I. 
\label{eq:gdee}
\eea
and the index $I$ enumerates linearly independent solutions not proportional to $l^\mu$. 
Note that we take into account only the two transverse components of the vectors. This is possible because   we may impose the condition $\xi^{\mathbf{ 0}} = 0$, or equivalently 
\bea
\xi^\mu\,l_\mu = 0, \label{eq:orthogonal}
\eea
along the whole $\gamma_0$,  and because the component $\xi^{\mathbf{ 3}}$  does not couple with the other three. $\xi^{\mathbf{3}}$ is also irrelevant from the point of view of geometric optics \cite{korzynski2018,grasso2019}. 
The cross-section of $S$ by the screen spanned by $e^\mu_\mba$ is spacelike and Lorentz-invariant everywhere along the geodesic. Especially important for the main result is its cross-sectional area:
\bea 
\calA = \pi\, \epsilon_{\mba \mbb} \,\xi^\mba_1\, \xi^\mbb_2,
\label{eq:ardef}
\eea 
where $\epsilon_{\mba\mbb}$ is the area two-form   \cite{sachs1961}  and $\xi^\mba_I$ are the two linearly independent solutions of \eqref{eq:gdee} projected on the screen space. 
The area defined by \eqref{eq:ardef} is a signed quantity which may change sign when the bundle degenerates to a point or to a line. We assume throughout the work that the initial value of $\calA$ is chosen to be positive, i.e.
\bea
\calA(\lambda_\calO) \equiv \calA_\calO > 0. \label{eq:Ainitialdata}
\eea

It is customary to rewrite the GDE \eqref{eq:gdee} for the two generators of an infinitesimal bundle in terms of the so-called kinematical bundle variables. We first need to decompose the transverse part of the optical tidal tensor into the trace and the traceless part as follows:
\bea
 R\UD{\mba }{\bm\mu\bm\nu\mbb }\,l^{\bm \mu}\,l^{\bm \nu} = -\frac{1}{2}\,R_{ll}\,\delta\UD{\mba }{\mbb } + C\UD{\mba }{\bm\mu\bm\nu\mbb}\,l^{\bm\mu}\,l^{\bm\nu}, \label{eq:decomposition}
\eea
where $R_{ll} = R_{\bm\mu\bm\nu}\,l^{\bm\mu}\,l^{\bm\nu}$ denotes the $l$-$l$ component of the Ricci tensor and $C\UD{\bm \alpha}{\bm\mu\bm\nu\bm \beta}$ is the Weyl tensor. This decomposition  holds even though we are taking the trace only with respect to the two-dimensional subspace of the tangent space, see for example \cite{grasso2019, schimd2005}. 

The bundle evolution along the null geodesic is most conveniently described in terms of the deformation tensor $B\UD{\mba}{\mbb}$ \cite{wald2010, poisson2004, galloway2004}, defined via
\bea
 \dot \xi_I^\mba(\lambda) = B\UD{\mba}{\mbb}(\lambda)\,\xi_I^\mbb(\lambda)
\eea
for both $I=1,2$. 
The infinitesimal bundle can always be extended to a full congruence of null geodesics. In that case we have the relation
\bea
 B\UD{\mba}{\mbb} = \nabla_{\mbb} l^{\mba},
 \eea
where $l^\mu$ denotes the vector field generating the congruence. 

$B\UD{\mba}{\mbb}$ decomposes into the kinematical variables (also known as the optical scalars), i.e. the scalar expansion $\theta$, traceless symmetric shear $\sigma_{\mba\mbb}$, and antisymmetric twist $\omega_{\mba\mbb}$, according to
\bea
B_{\mba\mbb} = \frac{1}{2}\theta\,\delta_{\mba\mbb} + \sigma_{\mba\mbb} + \omega_{\mba\mbb}.
\eea
Each of them satisfies an appropriate  ODE along the null geodesic, known as the null Raychaudhuri equations,  Sachs optical equations or transport equations.
Here we only consider  twist-free (or surface-forming) bundles, for which $\omega_{\mba \mbb} = 0$ along the whole $\gamma_0$. In   this case the equations for $\theta$ and shear $\sigma_{\mba \mbb}$ read \cite{poisson2004}
\bea
\dfrac{\dd \theta}{\dd \lambda} &=& - \frac{\theta^2}{2} - \sigma_{\mba \mbb} \sigma^{\mba \mbb} - R_{ll} 
\label{eq:raye} \\
\dfrac{\dd \sigma_{\mba \mbb}}{\dd \lambda} &=& - \theta \,\sigma_{\mba \mbb} + C_{\mba  ll \mbb},
\label{eq:rays}
\eea
where $C_{\mba ll \mbb} = C_{\mba\bm\mu\bm\nu\mbb}\,l^{\bm\mu}\,l^{\bm\nu}$. These equations are sometimes written in a different form, using a set of complex scalars, but we prefer the tensorial representation.
The area of the cross-section satisfies its own evolution equation, given in terms of the expansion:
\bea
\dfrac {\dd \calA} {\dd \lambda} = \calA\, \theta.
\label{eq:Aevolution}
\eea

In this paper we mainly consider two examples of infinitesimal ray bundles.  The first one is   the \emph{infinitesimal ray bundle parallel at $\calO$}, or \emph{the parallel bundle at $\calO$} for short, see Fig.~\ref{fig:initiallyparallel}. As the name suggests, it satisfies the condition of being strictly parallel at $\calO$, i.e.
\bea
 \theta(\lambda_\calO) &=& 0 \\
 \sigma_{\mba\mbb}(\lambda_\calO) &=& 0 \\
 \omega_{\mba\mbb}(\lambda_\calO) &=& 0 .
\eea

\begin{figure}[H]
\centering
\includegraphics[trim={1cm 9cm 2cm 5cm},clip,width=0.8\linewidth]{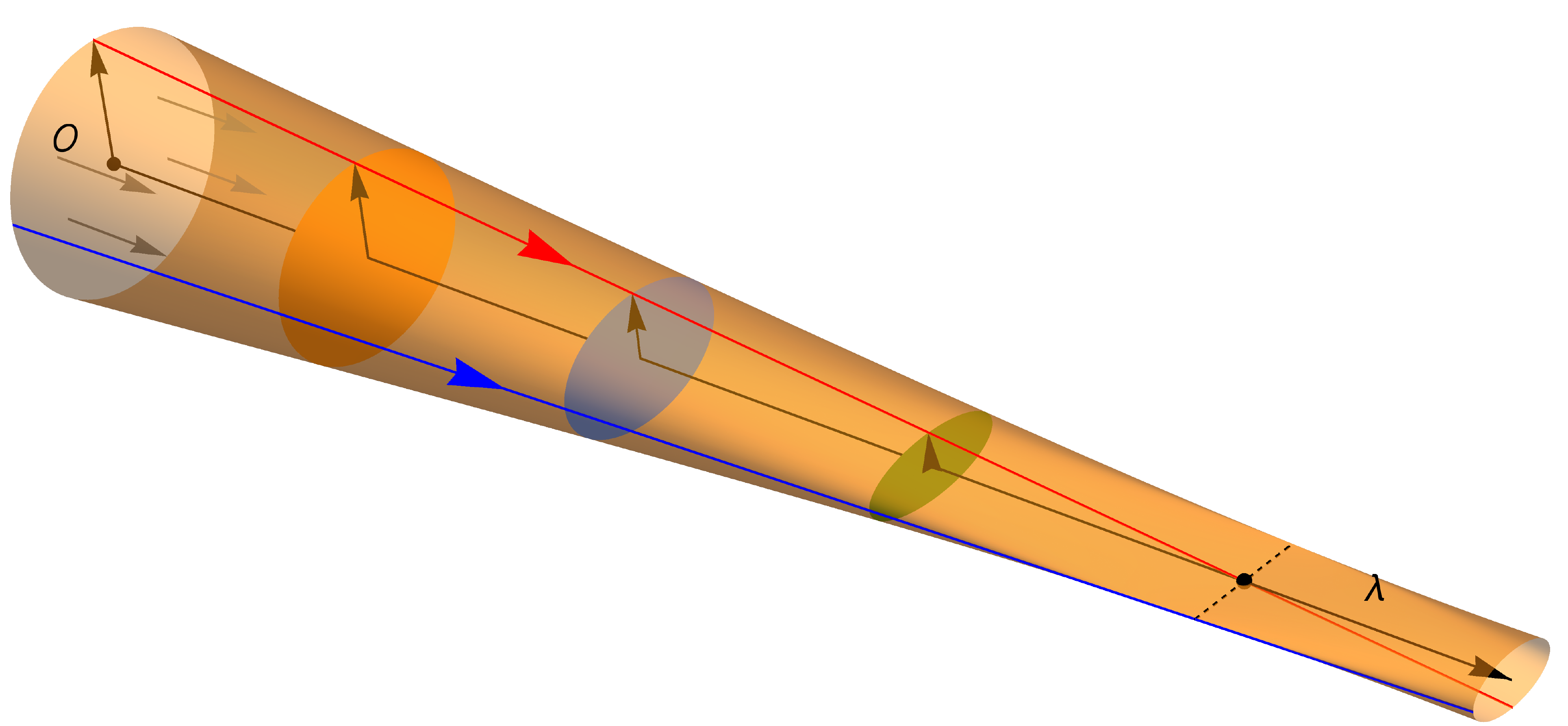}
\caption{
A bundle of rays that runs parallel at $\calO$. Its shape undergoes a deformation under the spacetime curvature. Along the fiducial geodesic a shape of cross section that is circular at $\calO$ changes its size and becomes increasingly elliptical. In the generic case it eventually degenerates either to a line or to a point. The point on the geodesic where this degeneracy happens is called the focal point.} \label{fig:initiallyparallel}
\end{figure}

It is related to the transverse components of  $\WXX$: namely, we have
\bea
\xi^{\mba}_I(\lambda) = {\WXX}\UD{\mba}{\mbb}(\lambda)\,\xi^{\mbb}_I(\lambda_\calO) \label{eq:xiWXX}
\eea
for this bundle. Due to the orthogonality condition (\ref{eq:orthogonal}) $\xi_I^\mu$ has only transverse components plus a component proportional to $l^\mu$. The latter is irrelevant
from the point of view of the geometry of cross-sections, see the Sachs shadow theorem \cite{sachs1961}, so it is the two transverse components of $\xi_I^\mu$ given by \eqref{eq:xiWXX} that define the 
distance measures.
Finally, Eqs. \eqref{eq:xiWXX} and \eqref{eq:ardef} give
\bea
\map \calA \lambda = \paren {\det {\WXX} \UD \mba \mbb} \,\calA_\calO. \label{eq:AfromWXX}
\eea
From this and \eqref{eq:muWXX} we see that
\bea
\mu = 1 - \frac{\map \calA \lambda}{\calA_\calO},
\label{eq:muar}
\eea
where $\mu$ is calculated for the emission point at $\lambda$ and the observation point at $\calO$.

The other bundle we consider in this paper is the \emph{infinitesimal ray bundle with a vertex at point $\calE$}, or simply the \emph{vertex bundle from $\calE$}. It satisfies the condition of vanishing at $\calE$, i.e. $\tilde\xi_I^{\mba}(\lambda_\calE) = 0$ for $I = 1,2$, see Fig.~\ref{fig:vertexatEa}.
Let $\tWXL{}\UD{\mba}{\mbb}(\lambda)$ denote the transverse components of the bitensor defined just like $\WXL{}\UD{\mba}{\mbb}(\lambda)$, but with the initial point taken at $\calE$ instead of $\calO$. The vertex bundle is then related to this bitensor via:
\bea
\tilde\xi_I^\mba(\lambda) = \tWXL{}\UD{\mba}{\mbb}(\lambda) \,\dot{\tilde\xi}_I^\mbb(\lambda_\calE).
\eea
The asymptotic behaviour of the infinitesimal bundle near $\calE$ is described by the Taylor expansion
\bea
\tWXL{}\UD{\mba}{\mbb}(\lambda) &=& (\lambda - \lambda_\calE)\,\delta\UD{\mba}{\mbb} + O\left(\left(\lambda - \lambda_\calE\right)^2\right), \label{eq:WXLasympt}
\eea
see equations \eqref{eq:WXLdef} with $\calE$ as the base point.
It follows that the expansion $\tilde \theta$ has a simple pole at $\lambda_{\calE}$ \cite{perlick2004r}:
\bea
\tilde \theta(\lambda) &=& 2(\lambda - \lambda_{\calE})^{-1} + O(1). \label{eq:thetaasympt}
\eea

\begin{figure}[H]
\begin{subfigure}[b]{1.\linewidth}
\centering
\includegraphics[trim={0cm 7cm 5cm 3cm},clip,width=0.6\linewidth]{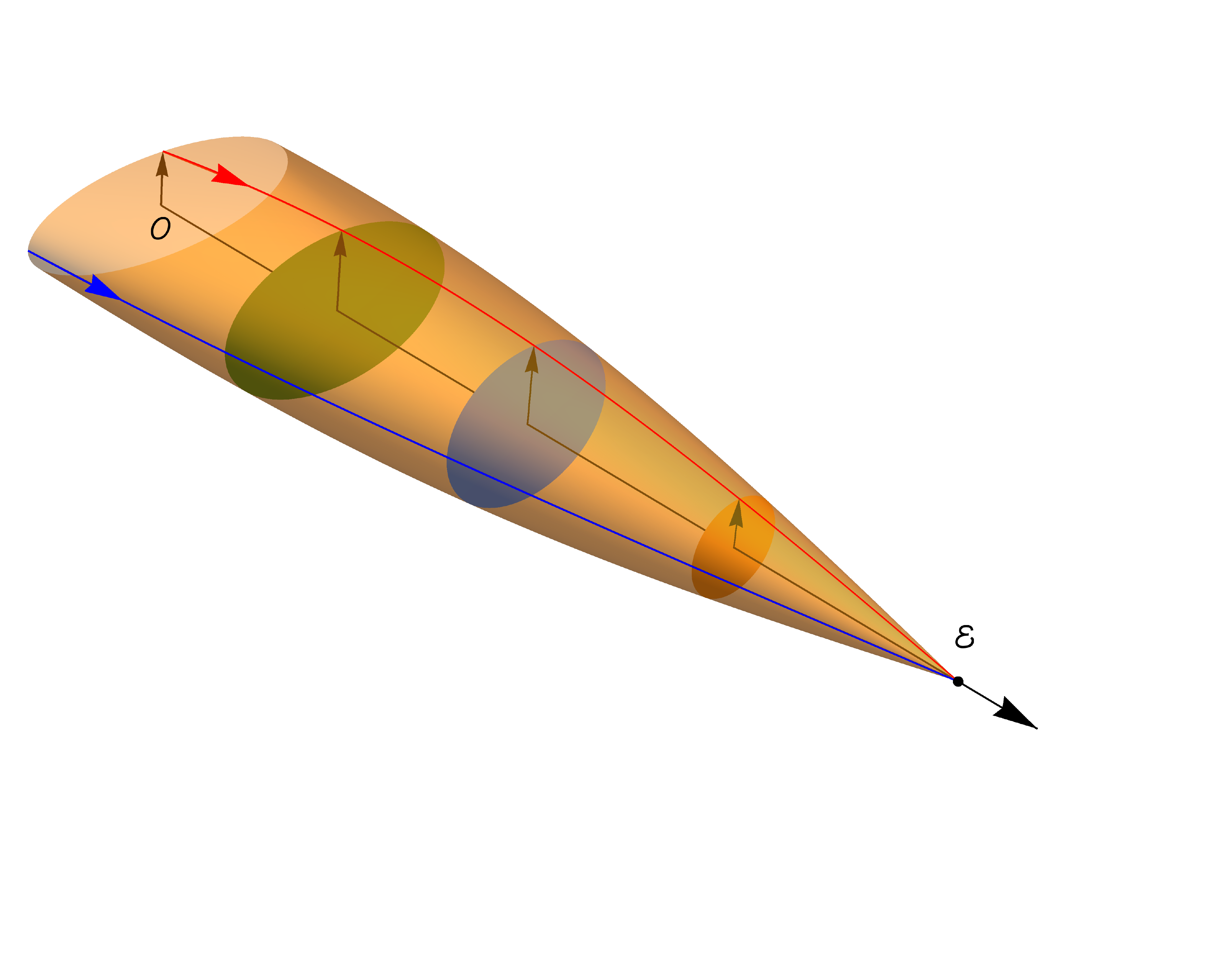}
\caption{}\label{fig:vertexatEa}
\end{subfigure}\\
\begin{subfigure}[b]{1.\linewidth}
\centering
\includegraphics[trim={5cm 6cm 5cm 5cm},clip,width=0.8\linewidth]{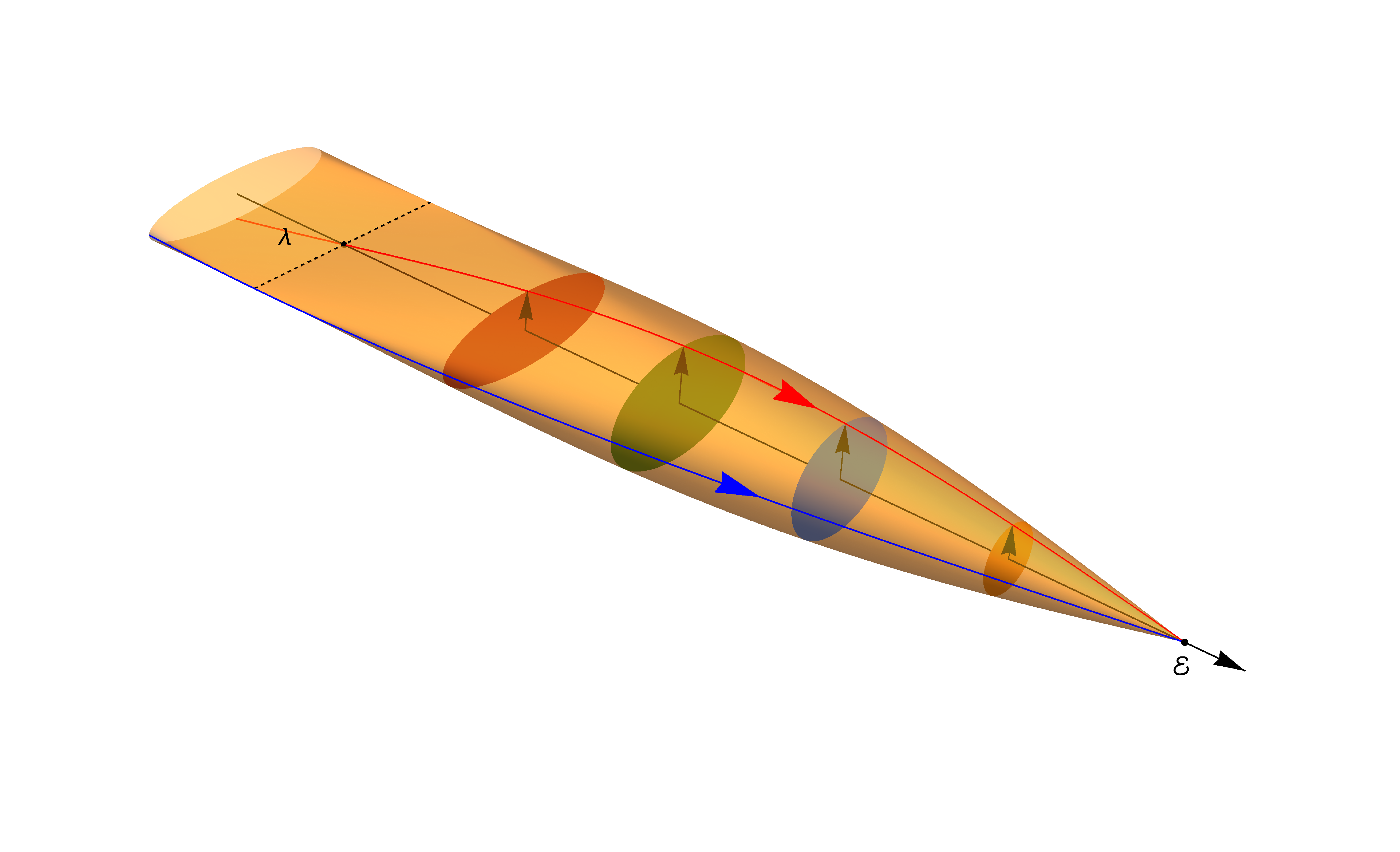}
\caption{}\label{fig:vertexatEb}
\end{subfigure}

\caption{(a) A bundle of rays that forms a vertex at $\calE$. Its shape undergoes a deformation under the spacetime curvature. Along the fiducial geodesic a shape of cross section that is circular at $\calE$ changes its size and becomes increasingly elliptical. (b) The same bundle of rays that forms a vertex at $\calE$. At a larger distance it may eventually degenerate either to a line or to a point. The point on the geodesic where this degeneracy happens is called the conjugate point.}
\label{fig:vertexatE}
\end{figure}

This is an example of a singular point of an inifinitesimal ray bundle. In the next section we will define this notion more precisely.

By analogy we may also consider an infinitesimal ray bundle with a vertex at $\calO$, related to the bitensor $\WXL$ instead of $\tWXL$.

\subsection{Singular points of a bundle}

The description of a ray bundle using the shear and expansion can break down at isolated points even though the perturbed geodesics constituting the bundle are perfectly regular there. This typically  happens when the  bundle collapses along one or two directions, forming a self-intersection. In the language of the Sachs frame this happens if the determinant of the two transverse solutions of the GDE vanishes, i.e. $\det (\xi^{\bm A}_1, \xi^{\bm B}_2) = 0$. The expansion diverges in this case to $\pm \infty$ and changes sign.

We introduce the following definition:
\begin{ddefinition} Point $\cal P$ along a null geodesic $\gamma_0$ is a \emph{ regular point of a ray bundle} iff $\theta(\lambda)$, $\omega_{\bm A\bm B}(\lambda)$ and $\sigma_{\bm A \bm B}(\lambda)$  are smooth  at $\cal P$. A point that is not regular will be called \emph{singular}. \end{ddefinition}

We fix the null geodesic $\gamma_0$ and the observer's position $\calO$ along $\gamma_0$. We can now introduce two types of singular points along a  null geodesic, defined by the properties of the vertex and initially parallel ray bundles. 

We call $\cal P$ a \emph{conjugate point} with respect to $\calO$ iff the vertex bundle from $\calO$ refocuses back at
$\cal P$ at least along one transverse direction. This property is equivalent  to the existence of a  Jacobi field along $\gamma_0$, satisfying the GDE and vanishing at $\cal O$ and $\cal P$, but not identically zero \cite{perlick2004r}. It is easy to check that this happens iff $\det {\WXL}\UD{\bf A}{\bf B} = 0$ between $\calO$ and $\calP$. Conjugate points correspond to the intersection of the fiducial geodesic with a caustic and are points of infinite magnification
of images of objects located at $\cal P$ as seen in $\calO$. 

We call $\cal P$ a \emph{focal point} iff an infinitesimal bundle of rays running parallel at $\cal O$ refocuses at $\cal P$ along at least one direction.
This happens when $\det  {\WXX}\UD{\bf A}{\bf B} = 0$, see Fig \ref{fig:initiallyparallel}. 

The Reader may check that at both the focal point and the conjugate point the expansion $\theta$ of the corresponding bundle has a singularity.

\subsection{Integral formula for the cross-sectional area}

The ODE (\ref{eq:Aevolution}) can be integrated exactly assuming that $\theta(\lambda)$ is a regular function. Namely, if there are no singular points between $\calO$ and $\lambda$, 
the solution simply reads
\bea
 \calA(\lambda) = \calA_{\calO} \exp\left(\int_{\lambda_{\cal O}}^\lambda \theta(\lambda')\,\dd\lambda'\right) \label{eq:Asolution}.
\eea
This formula will play an important role in the proof of the main result. Let us stress that the condition of regularity of $\theta(\lambda)$ is crucial here, because Eq. (\ref{eq:Asolution}) may break down after  a singular point such as  a focal point. This is evident if we note that $\calA(\lambda)$ may switch sign  past a focal point, which is obviously inconsistent with Eq. (\ref{eq:Asolution}), in which the signs of $\calA_\calO$ and $\calA(\lambda)$ must be the same.

\section{The main theorems}

\begin{definition}
We say that the  the \emph{null energy condition} (NEC) holds at a set $O \subset M$ iff at every point in $O$ we have $R_{\mu\nu}\,l^\mu\,l^\nu \ge 0$ for all null vectors $l^\mu$.
\end{definition}

\begin{theorem} \label{thm:mu}
Let $\calO$ and $\calE$ be two points along a null geodesic $\gamma_0$ such that $\calO$ lies in the causal future of $\calE$ and let the NEC hold  along $\gamma_0$ between $\calO$ and $\calE$. Assume also that between $\calO$ and $\calE$  there are no singular points of the infinitesimal bundle of rays parallel at $\calO$. Then we have
\begin{equation}
 \mu \ge 0
\end{equation}
for an observer at $\calO$ and a source at $\calE$.
Moreover, $\mu=0$ iff the transverse components of the optical tidal tensor $R\UD{\bm A}{\bm \mu \bm \nu \bm B}\,l^{\bm \mu}\,l^{\bm \nu}$ vanish along $\gamma_0$ between $\calO$ and $\calE$.
\end{theorem}
\begin{proof}
We begin with the inequality. The right-hand side of Eq. (\ref{eq:raye}) for the derivative  of $\theta$ is obviously non-positive. Since initially $\theta(\lambda_{\calO}) = 0$ we see that  $\theta(\lambda) \le 0$ everywhere
between $\calO$ and $\calE$:  
\bea
\int_{\lambda_\calO}^{\lambda_{\calE}} \theta(\lambda)\,\dd\lambda \le 0.
\eea
From the integral formula (\ref{eq:Asolution}) we see that $\calA(\lambda) \le \calA_{\calO}$, so from \eqref{eq:muar} we have $\mu \ge 0$. This completes the proof of the first part of the theorem.

Assume now $\mu = 0$ between $\calO$ and $\calE$. It follows from \eqref{eq:muar} that $\calA(\lambda) = \calA_{\calO}$. Substituting this to the integral formula (\ref{eq:Asolution}) we obtain
\bea
\exp\left(\int_{\lambda_\calO}^{\lambda_\calE} \theta(\lambda)\,\dd\lambda \right) = 1,
\eea
or equivalently
\bea
\int_{\lambda_\calO}^{\lambda_\calE} \theta(\lambda)\,\dd\lambda =0. \label{eq:int0}
\eea
By the regularity assumption $\theta(\lambda)$ is continuous on the closed interval between $\lambda_\calO$ and $\lambda_\calE$ and we have also shown that $\theta(\lambda) \le 0$ on this interval. Therefore Eq. \eqref{eq:int0} is only possible if
$\theta(\lambda) = 0$ everywhere on this interval. Substituting this condition to (\ref{eq:raye}) we obtain $-\sigma_{\bm A\bm B}\,\sigma^{\bm A\bm B} - R_{ll} = 0$, which implies that both
$\sigma_{\bm A\bm B} = 0$ and $R_{ll} = 0$ everywhere between $\calO$ and $\calE$. Finally, we substitute the former relation to (\ref{eq:rays}) to obtain $C_{\bm A\bm \mu\bm \nu\bm B}\,l^{\bm \mu}\,l^{\bm \nu} = 0$. 

We have thus proved that both the contracted Ricci tensor and the transverse components of the contracted Weyl tensor vanish. It follows that that all transverse components of the optical tidal tensor must vanish between $\calO$ and $\calE$, see \eqref{eq:decomposition}. This completes the proof of the second part of the theorem.
\end{proof}
This theorem does not automatically imply an inequality between the distance measures because of the sign ambiguity in the definition of $\mu$, see the rhs of Eq. (\ref{eq:muD}). We therefore need one more result regarding the sign-defining factor $\sigma$ in \eqref{eq:muD}:
\begin{proposition} \label{pro:sigma}
Under the assumptions of Theorem \ref{thm:mu} we have $\sigma = 1$.
\end{proposition}
\begin{proof}
From (\ref{eq:Asolution}) we see easily that $\calA(\lambda) > 0$. Therefore from (\ref{eq:AfromWXX}) we have $\det {W_{XX}}\UD{\bf A}{\bf B} > 0$. Then $\sigma = 1$ follows from (\ref{eq:sigmaWXX}).
\end{proof}
The following result follows now Theorem \ref{thm:mu}, equation (\ref{eq:muD}) and Proposition \ref{pro:sigma} 
\begin{corollary} \label{thm:D}
Under the assumptions of Theorem \ref{thm:mu} we have
\begin{equation}
 D_{par} \ge D_{ang}
\end{equation}
for any observer at $\cal O$ and any light source at $\cal E$. Moreover, $D_{par} = D_{ang}$ iff the transverse components of the optical tidal tensor $R\UD{\bm A}{\bm \mu \bm \nu \bm B}\,l^{\bm \mu}\,l^{\bm \nu}$ vanish along $\gamma_0$ between $\calO$ and $\calE$.
\end{corollary}
\begin{proof} From Proposition \ref{pro:sigma}  and \eqref{eq:muD} we have $\mu = 1-\frac{D_{ang}^2}{D_{par}^2}$. Together with the positivity of $D_{ang}$ and $D_{par}$, this implies that
the inequality $\mu \ge 0$ is equivalent to $D_{par} \ge D_{ang}$ and the equality $\mu = 0$ is equivalent to $D_{par} = D_{ang}$. The Corollary follows then trivially from Theorem \ref{thm:mu}.
\end{proof}

\subsection{Trace-based baseline-averaged parallax distance}

In \textcite{rasanen2014}, \textcite{rosquist1988} as well as  \textcite{ellis1971} (republished as \cite{ellis2009rep}), \textcite{jordan1961} (republished as \citep{jordan2013rep}) a different method of baseline averaging for the parallax distance has been proposed. Effectively it differs from $D_{par}$ by using the trace of the parallax matrix instead of the determinant for baseline direction averaging.
We will now prove a result analogous to Theorem \ref{thm:mu} and Corollary \ref{thm:D} for the parallax distance averaged by trace. 
Let $\tilde D_{par}$ denote the trace-based parallax distance:
\bea
\tilde D_{par} = \frac{2}{\left|\Pi\UD{\mba}{\mba}\right|}.
\eea
This can be expressed also as
\bea
\tilde D_{par} = 2\,\left(l_{{\cal O}\,\mu}\,u_{\cal O}^\mu\right)\big|\tilde\theta\big|^{-1}, \label{eq:tildeDparfromtheta}
\eea
where $\tilde \theta$ is the expansion of the vertex bundle emanating from $\cal E$, evaluated at $\cal O$, see Fig. \ref{fig:vertexatEa}. 
If we rescale the parametrization $\lambda$ to fit the observer's frame, ensuring $l_{\calO\,\mu}\,u_\calO^\mu = 1$, this expression simplifies to
\bea
\tilde D_{par} = 2\,\big|\tilde\theta\big|^{-1}.
\eea

We prove now the following theorem:
\begin{theorem} \label{thm:tildeD}
Let $\calO$ and $\calE$ be two points along a null geodesic $\gamma_0$ such that $\calO$ lies in the causal future of $\calE$ and let the NEC hold  along $\gamma_0$ between $\calO$ and $\calE$. Assume also that between $\calO$ and $\calE$  there are no singular points of the infinitesimal bundle of rays with the vertex at $\cal E$, except the point $\cal E$ itself, and that its expansion $\tilde \theta$ does not vanish between $\cal O$ and $\cal E$. Then we have
\begin{equation}
\tilde D_{par} \ge D_{ang}
\end{equation}
for an observer at $\calO$ and a source at $\calE$.
Moreover, $\tilde D_{par} = D_{ang} $ iff the transverse components of the optical tidal tensor $R\UD{\bm A}{\bm \mu \bm \nu \bm B}\,l^{\bm \mu}\,l^{\bm \nu}$ vanish along $\gamma_0$ between $\calO$ and $\calE$.
\end{theorem}
\begin{proof}
The proof uses the propreties of an infinitesimal bundle with a vertex at $\cal E$ in a similar way as the proof of Theorem \ref{thm:mu} uses the infinitesimal bundle parallel at $\calO$. 
We begin by relating the angular diameter distance to the properties of this bundle.

The angular diameter distance is related to the determinant of the Jacobi matrix at $\calE$, with the vertex positioned at $\calO$, see \eqref{eq:distDang}. However, it can be easily related to the Jacobi map
with the roles of $\calO$ and $\calE$ reversed. 

Let $\tWXL{}\UD{\mba}{\mbb}(\lambda)$ denote the Jacobi map with the vertex at $\calE$, i.e. satisfying
\bea
\frac{d^2}{d\lambda^2} \tWXL {}\UD{\mba}{\mbb} - R\UD{\mba}{ll\mbc}\,\tWXL{}\UD{\mbc}{\mbb} &=& 0 \\
{\tWXL}{} \UD{\mba}{\mbb}(\calE) &=& 0 \\
\frac{d}{d\lambda}\tWXL{}\UD{\mba}{\mbb}(\calE) &=& \delta\UD{\mba}{\mbb}.
\eea
From the symplectic property of the GDE we have a simple relation between the Jacobi matrix from $\calE$ up to $\calO$ and the one calculated the other way round:
\bea
\WXL{}\UD{\mba}{\mbb}(\lambda_\calE) = -\tWXL{}\DU{\mbb}{\mba}(\lambda_\calO),
\eea
i.e. the two matrices are the  transpose of each other, with a sign flip \cite{perlick2004r, uzun2020}.
Therefore we can replace $\det \WXL{}\UD{\mba}{\mbb}$ by $\det \tWXL{}\UD{\mba}{\mbb}$ in \eqref{eq:distDang}:
\bea
D_{ang} = (l_{\calO\,\mu}\,u_{\cal O}^\mu)\,\left|\det  \tWXL{}\UD{\mba}{\mbb}\right|^{-1/2}, \label{eq:DangfromtWXL}
\eea 
where $\tWXL{}\UD{\mba}{\mbb}$ is the Jacobi map from $\calE$ to $\calO$. \footnote{The possibility of replacing the determinant of the Jacobi matrix calculated from $\calO$ to $\calE$  by the one with the endpoints reversed is also the key step in the proof of the Etherington's duality relation between the angular diameter distance and the luminosity distance, see  \cite{uzun2020}.}

$\tWXL{}\UD{\mba}{\mbb}(\lambda)$, on the other hand, can be related by another ODE to the deformation tensor $\tilde B\UD{\mba}{\mbb}$ of the bundle with vertex at $\calE$:
\bea
\frac{d}{d\lambda} \tWXL{}\UD{\mba}{\mbb} = \tilde B\UD{\mba}{\mbc}\,\tWXL{}\UD{\mbc}{\mbb}, \label{eq:dtildeD}
\eea
The bundle  is twist-free, so it decomposes into expansion $\tilde\theta$ and shear $\tilde\sigma_{\mba\mbb}$ according to
$\tilde B \UD{\mba}{\mbb}(\lambda) = \frac{1}{2}\tilde\theta \,\delta\UD{\mba}{\mbb} + \tilde \sigma\UD{\mba}{\mbb}$.  It follows from \eqref{eq:dtildeD} that
\bea
\frac{d}{d\lambda} \det \tWXL{}\UD{\mba}{\mbb} = \tilde\theta\,\left(\det \tWXL{}\UD{\mba}{\mbb}\right). \label{eq:ddettildeD}
\eea

We now define an auxilliary function $f(\lambda)$ in the following way: we fix the emission point $\calE$ and vary the observation point, corresponding to the affine parameter value $\lambda$. We then take the ratio of the distances
squared, measured between $\calE$ and the point $\lambda$:
\bea
f(\lambda) = \frac{D_{ang}^2}{\tilde D_{par}^2}.
\eea
From \eqref{eq:DangfromtWXL} and \eqref{eq:tildeDparfromtheta} we get an expression for $f(\lambda)$ in terms of quantities related to the vertex bundle at $\calE$: 
\bea
 f(\lambda) = \frac{\tilde\theta^2\,\left|\det \tWXL{}\UD{\mba}{\mbb}\right|}{4}. \label{eq:fdef}
\eea
We will now derive an integral formula for $f(\lambda)$. We begin by differentiating Eq. (\ref{eq:fdef}) and applying (\ref{eq:ddettildeD}):
\bea
 \frac{df}{d\lambda} = \frac{1}{4}\left(2\tilde\theta\,\frac{d\tilde\theta}{d\lambda}\,\left|\det \tWXL{}\UD{\mba}{\mbb}\right|+\tilde\theta^3\,\left|\det \tWXL{}\UD{\mba}{\mbb}\right|\right)
\eea
In the first term use the propagation equation \eqref{eq:raye} for $\tilde\theta$  to obtain
\bea
\frac{df}{d\lambda} = \frac{1}{4}\tilde\theta\,\left|\det \tWXL{}\UD{\mba}{\mbb}\right|\left(-\tilde\sigma_{\mba\mbb}\,\tilde\sigma^{\mba\mbb}-R_{ll}\right),
\eea
or, equivalently,
\bea
\frac{df}{d\lambda} = \frac{f}{\tilde\theta}\left(-\tilde\sigma_{\mba\mbb}\,\tilde\sigma^{\mba\mbb}-R_{ll}\right). \label{eq:ODEf}
\eea
This equation can be solved by separation of variables, but we need the initial data. 
Recall that in the bundle with the vertex at $\calE$ we have asymptotic expansions \eqref{eq:WXLasympt} and \eqref{eq:thetaasympt}.
The reader may check using \eqref{eq:fdef} that this implies that $f \to 1$ as $\lambda \to \lambda_\calE$. With this knowledge we may integrate
the ODE \eqref{eq:ODEf} to
\bea
f(\lambda) = \exp\left(-\int_{\lambda_\calE}^\lambda\tilde\theta^{-1}\,\left(\tilde\sigma_{\mba\mbb}\,\tilde\sigma^{\mba\mbb}+R_{ll}\right)\,\dd\lambda'\right).
\eea
The integrand contains the expansion $\tilde\theta$ in the denominator, but the integral is regular at $\calE$ and everywhere along the interval considered because of our assumptions regarding the behaviour of $\tilde\theta$ (i.e. no zeros and a pole at $\calE$). Moreover, 
 we note that $\tilde\theta < 0$ inside the interval we consider, because it is negative near $\lambda_\calE$ due to \eqref{eq:thetaasympt} and from the assumptions we know that it cannot vanish or change sign in the interval we consider.   It also follows  that we can write $\tilde\theta = -\big|\tilde\theta\big|$.  The integration in the formula above proceeds from larger $\lambda_\calE$ down to smaller $\lambda$, so  we swap the integration limits and absorb this way the minus sign from $\tilde\theta$ to obtain:
\bea
f(\lambda) = \exp\left(-\int_{\lambda}^{\lambda_\calE}\big|\tilde\theta\big|^{-1}\,\left(\tilde\sigma_{\mba\mbb}\,\tilde\sigma^{\mba\mbb} +R_{ll}\right)\,\dd\lambda'\right).
\eea

Now, the last steps of the proof proceed just like in the proof of Theorem \ref{thm:mu}: the integrand is manifestly non-negative if NEC holds. Therefore we have $f \le 1$, with the equality happening only if and only if 
both $R_{ll}$ and $\tilde \sigma_{\mba\mbb}$ vanish between $\calE$ and $\lambda$. The vanishing of the shear tensor $\tilde \sigma_{\mba\mbb}$ implies that the transverse components of the Weyl part of the optical tidal matrix $C_{\mba ll \mbb}$ must also vanish  because of \eqref{eq:rays}.

\end{proof}

\subsection{Single baseline parallax distance}
\label{sec:sgbase}

While it may not be obvious from the proofs we presented above, the baseline averaging of the trigonometric parallax effect is necessary for the distance inequality to work. Obviously it is possible to define a type of parallax distance defined through measurements with a single baseline. Let the transverse unit vector   $n^\mba$ denote the baseline direction, with $n^\mba\,n_\mba = 1$. As the simplest example we consider here
\bea
D_{par}^n = \left|\Pi_{\mba\mbb}\,n^\mba\,n^\mbb\right|^{-1},
\eea
definition considered for example in \cite{kasai1988}. $D_{par}^n$ coincides with the baseline-averaged parallax distances in flat spacetimes, in which $\Pi\UD{\mba}{\mbb} = D^{-1}\,\delta\UD{\mba}{\mbb}$, $D$ denoting the distance in the observer's frame. However, it turns out $D_{par}^n$ does not have to obey the inequality $D_{ang} \le D_{par}^n$ even if the NEC is satisfied. As an example consider the situation when we have a vacuum solution ($R_{ll}=0$), but non-vanishing Weyl tensor $C_{{\mba}ll{\mbb}}$  causing shear of null geodesics along $\gamma_0$. Obviously the NEC holds in a vacuum spacetime. The most general Taylor expansions for ${W_{XX}}\UD{\mba}{\mbb}$ and ${W_{XL}}\UD{\mba}{\mbb}$ around $\lambda_\calO$ read \cite{serbenta2021}:
\bea
{W_{XX}}\UD{\mba}{\mbb} &=& \delta\UD{\mba}{\mbb} + \frac{(\lambda-\lambda_\calO)^2}{2}\,R\UD{\mba}{ll\mbb} + O\left((\lambda-\lambda_\calO)^3\right) \\
{W_{XL}}\UD{\mba}{\mbb} &=&(\lambda-\lambda_\calO) \,\delta\UD{\mba}{\mbb} + \frac{(\lambda-\lambda_\calO)^3}{6}\,R\UD{\mba}{ll\mbb} + O\left((\lambda-\lambda_\calO)^4\right),
\eea
with $R\UD{\mba}{ll\mbb}$ evaluated at $\calO$.
It follows that the parallax matrix has the Taylor expansion
\bea 
\Pi_{\mba\mbb} = (l_{\calO\,\mu}\,u_{\calO}^\mu)^{-1}\,(\lambda-\lambda_\calO)^{-1}\left(\delta_{\mba\mbb} + \frac{(\lambda-\lambda_\calO)^2}{3}R_{\mba ll \mbb}\right) + O\left((\lambda-\lambda_\calO)^2\right),
\eea
while the expansion for $D_{par}^n$ reads
\bea
D_{par}^n = (l_{\calO\,\mu}\,u_{\calO}^\mu)\,(\lambda-\lambda_\calO)\,\left(1 - \frac{(\lambda-\lambda_\calO)^2}{3}\,R_{\mba ll \mbb}\,n^\mba\,n^\mbb\right) + O\left((\lambda-\lambda_\calO)^4\right).
\eea
We compare the latter expression with the Taylor series for $D_{ang}$:
\bea
D_{ang} =  (l_{\calO\,\mu}\,u_{\calO}^\mu)\,(\lambda-\lambda_\calO)\,\left(1 - \frac{(\lambda-\lambda_\calO)^2}{12}\,R_{ll}\right) + O\left((\lambda-\lambda_\calO)^4\right).
\eea
Assuming vacuum we have $R_{ll} = 0$ and $R_{\mba ll \mbb} = C_{\mba ll \mbb}$ and hence 
\bea 
D_{par}^n=D_{ang}\,\left( 1 - \frac{(\lambda-\lambda_\calO)^2}{3}\,C_{\mba ll \mbb}\,n^\mba\,n^\mbb\right) + O\left((\lambda-\lambda_\calO)^3\right).
\eea
Obviously  we have $D_{par}^n < D_{ang}$  near $\calO$ provided that  $C_{\mba ll \mbb}\,n^\mba\,n^\mbb > 0$, i.e. $n^\mba$ is chosen such that the images undergo stretching by the tidal forces in its direction.

\section{Conclusions}

The two main results of this paper, i.e. Theorems \ref{thm:D} and \ref{thm:tildeD}, mean that any observation of a systematic difference between the angular diameter distance and the baseline-averaged parallax distance with $D_{ang} > D_{par}$ would be difficult to reconcile with general relativity and the theory of light propagation as we understand them today. In particular,  the effects of  shear of the ray bundle due to tidal fields along the line of sight could not explain away a result of this kind. One would have to give up either the null energy condition or the assumption that light travels along null geodesics. Therefore, systematic comparison of both distance measurements may be considered an experimental test of the null energy condition, assuming that the general relativity and the geometric optics approximation for light propagation hold. Since $D_{ang}$ is related to $D_{lum}$ and the redshift through the Etherington's reciprocity relation, it is in principle possible to perform the measurement of distance slip $\mu$ using standard candles for which we have also precise redshift and trigonometric parallax measurements.

We note, however, that the measurements of the difference between the two distance measures seem impossible today, because the annual parallax effects are too small over the distances in which we can measure the trigonometric parallax.  We can provide an order-of-magnitude estimate of the distance slip  from the integral formula \eqref{eq:relation}: for the mass density comparable with the mass density in the thin disc of the Milky Way $\rho = 1 \,{\rm M}_{\odot}\,{\rm pc}^{-3}$ \cite{mcmillan2011}, negligible pressure terms in $T_{\mu\nu}$ and the distance of 20 kpc, comparable to the largest distance measured by trigonometric parallax \cite{sanna2017}, we get $\mu \approx 2 \cdot 10^{-4}$. A successful measurement of $\mu$ for a single source would then require the determination of both distance measures and the redshift with  relative error not greater than $10^{-4}$, way below current limitations \cite{reid2014, paragi2014, rioja2020}.  This problem could possibly be overcome with a sufficiently large sample of sources. However, it would still require very precise standartization of the standard candles, good control of all possible sources of systematics as well as very precise redshift measurements. 

Another type of parallax measurement has been proposed by Kardashev \cite{kardashev}.  The measurement would use  the displacement due to the motion of the Solar System with respect to the CMB 
rest frame as the baseline. The baseline grows in this case linearly in time and the signal is measured as secular variations of angular separations between distant sources.  Longer baseline should in principle allow for parallax measurements on cosmological distances, although the signal seems still too low for modern instruments. Moreover, the foreground signal due to the Galactic aberration drift needs to be removed first \cite{Marcori:2018cw}.  The idea of cosmic parallax was also developed by many other authors \cite{mccrea1935, kardashev1973, novikov1977, novikov1978, rosquist1988, kasai1988, rasanen2014, Quercellini:2010zr, Marcori:2018cw, korzynski2020}.
Interestingly, the  estimates for the distance slip on cosmological distances in the standard $\Lambda$CDM model (satisfying the NEC)
yield fairly large values, reaching for example $\mu = 0.2$ near $z=1$ \cite{korzynski2020}.  Measurements  of $\mu$ on such distances and determination its sign could test for the NEC violation by dark energy, an obvious sign of physics beyond the  $\Lambda$CDM  model.

We also point out three caveats regarding the distance inequality. The first two are related to the limitations of the mathematical approach. In the proofs we have used the first order geodesic deviation equation around a null geodesic. This means that we assume that the linear, curvature term in the geodesic deviation equation describes very well the behaviour of all relevant light rays. This may fail, for example, if light passes through a region of very quickly varying gravitational potential accross the null ray bundle considered (physical width of around 1 AU). In particular, it may fail in case of a microlensing event, i.e., a small massive body passing through the line of sight. It can also fail if light rays undergo significant nongravitational bending, for example due to the presence of ionized medium of variable density along the line of sight, or if the geometric optics approximation is not applicable.

Second, the inequality works only up to the first focal point, whose position in a given direction is not known beforehand. However, focal points between the Earth and galactic sources should be very rare, confined to rather special, fairly strongly lensed rays. Assuming that the line of sight is filled uniformly with mass density of 100 ${\rm M}_{\odot}\,{\rm pc}^{-3}$, scale of the density inside the bulge  of the Milky Way, and ignoring the Weyl tensor contribution, we may predict the first focal point to appear at around 140 kpc. For uniform mass density comparable with the thin disc (1 ${\rm M}_{\odot}\,{\rm pc}^{-3}$) the distance to the focal point grows to over 1 Mpc, while the density of the dark matter halo of $10^{-2} \,{\rm M}_{\odot}\,{\rm pc}^{-3}$ yields 14 Mpc, the mass density estimates taken  again from \cite{mcmillan2011}. The assumption of uniform mass density along the light of sight makes these estimates very conservative. We conclude that we should not expect the formation of such points anywhere around galactic distances.
 Moreover, since ${W_{XX}}\UD{\mba}{\mbb}$ usually changes sign of the determinant at the focal point, we may expect the parallax matrix $\Pi\UD{\mba}{\mbb}$ past the focal point to deviate far from proportionality to the unit matrix. This in turn may lead to unusual dependence of the two-dimensional parallax angles on the Earth's position, an effect that is in principle observable. Therefore sources past focal point could in principle be detected and removed from the data. 

Thirdly, as we have seen in Sec. \ref{sec:sgbase}, if the parallax is measured only along one baseline then the test is prone to errors in the presence of strong shear induced by the Weyl tensor. Therefore, if shear is not negligible then the baseline averaging step is crucial: we must be able to measure the parallax in two orthogonal directions for the measurement to yield a reliable NEC test.  In case of the annual parallax due to Earth's motion this requirement excludes sources too close to the ecliptic. The problem is obviously even more serious for the cosmic parallax, in which only one baseline is available for all sources. Tests on cosmological distances require therefore prior estimation of the Weyl tensor contribution to the parallax matrix over large distances. Note, however, that
even in the presence of moderate shear this problem might in principle be overcome if the single-baseline parallax and angular diameter distance have been measured for a sufficiently large sample of sources at different positions on the sky. Since the Weyl tensor along the line of sight depends on the position of the source, it  varies relatively quickly accross the sky and it is uncorrelated with the fixed  baseline direction. We may expect that its impact will average to 0, at least in the linear order given a large sample of sources. What is left from the averaging is thus the bare, baseline-independent effect of the Ricci tensor. Quantification of the impact of  shear on the NEC test, its error budget and the question of feasibility are beyond the scope of this paper.
We stress here that unlike the full baseline-averaged measurement, this type of single-baseline measurement relies on additional assumptions about the metric tensor, i.e. either the vanishing of the Weyl contribution to the parallax matrix or the random, uncorrelated nature of the Weyl tensor over the whole sky.

\section*{Acknowledgments}
The work was supported by the National Science Centre, Poland (NCN) via the SONATA BIS programme, Grant No. 2016/22/E/ST9/00578 for the project \emph{``Local  relativistic  perturbative  framework  in  hydrodynamics  and  general relativity  and  its  application  to  cosmology"}.

\bibliography{main}

\begin{thebibliography}{51}
\expandafter\ifx\csname natexlab\endcsname\relax\def\natexlab#1{#1}\fi
\expandafter\ifx\csname bibnamefont\endcsname\relax
  \def\bibnamefont#1{#1}\fi
\expandafter\ifx\csname bibfnamefont\endcsname\relax
  \def\bibfnamefont#1{#1}\fi
\expandafter\ifx\csname citenamefont\endcsname\relax
  \def\citenamefont#1{#1}\fi
\expandafter\ifx\csname url\endcsname\relax
  \def\url#1{\texttt{#1}}\fi
\expandafter\ifx\csname urlprefix\endcsname\relax\def\urlprefix{URL }\fi
\providecommand{\bibinfo}[2]{#2}
\providecommand{\eprint}[2][]{\url{#2}}

\bibitem[{\citenamefont{Perlick}(2004)}]{perlick2004r}
\bibinfo{author}{\bibfnamefont{V.}~\bibnamefont{Perlick}},
  \bibinfo{journal}{Living Reviews in Relativity} \textbf{\bibinfo{volume}{7}},
  \bibinfo{pages}{9} (\bibinfo{year}{2004}).

\bibitem[{\citenamefont{Grasso et~al.}(2019)\citenamefont{Grasso,
  Korzy{\'n}ski, and Serbenta}}]{grasso2019}
\bibinfo{author}{\bibfnamefont{M.}~\bibnamefont{Grasso}},
  \bibinfo{author}{\bibfnamefont{M.}~\bibnamefont{Korzy{\'n}ski}},
  \bibnamefont{and} \bibinfo{author}{\bibfnamefont{J.}~\bibnamefont{Serbenta}},
  \bibinfo{journal}{Physical Review D} \textbf{\bibinfo{volume}{99}},
  \bibinfo{pages}{064038} (\bibinfo{year}{2019}).

\bibitem[{\citenamefont{Etherington}(1933)}]{etherington1933}
\bibinfo{author}{\bibfnamefont{I.}~\bibnamefont{Etherington}},
  \bibinfo{journal}{The London, Edinburgh, and Dublin Philosophical Magazine
  and Journal of Science} \textbf{\bibinfo{volume}{15}}, \bibinfo{pages}{761}
  (\bibinfo{year}{1933}).

\bibitem[{\citenamefont{Etherington}(2007)}]{etherington2007}
\bibinfo{author}{\bibfnamefont{I.}~\bibnamefont{Etherington}},
  \bibinfo{journal}{General Relativity and Gravitation}
  \textbf{\bibinfo{volume}{39}}, \bibinfo{pages}{1055} (\bibinfo{year}{2007}).

\bibitem[{\citenamefont{Penrose}(1966)}]{penrose1966}
\bibinfo{author}{\bibfnamefont{R.}~\bibnamefont{Penrose}}, in
  \emph{\bibinfo{booktitle}{Perspectives in Geometry and Relativity: Essays in
  honor of V{\'a}clav Hlavat{\'y}}} (\bibinfo{publisher}{Indiana University
  Press}, \bibinfo{year}{1966}), pp. \bibinfo{pages}{259--274}.

\bibitem[{\citenamefont{Kristian and Sachs}(1966)}]{kristian1966}
\bibinfo{author}{\bibfnamefont{J.}~\bibnamefont{Kristian}} \bibnamefont{and}
  \bibinfo{author}{\bibfnamefont{R.~K.} \bibnamefont{Sachs}},
  \bibinfo{journal}{The Astrophysical Journal} \textbf{\bibinfo{volume}{143}},
  \bibinfo{pages}{379} (\bibinfo{year}{1966}).

\bibitem[{\citenamefont{Ellis et~al.}(1971)}]{ellis1971}
\bibinfo{author}{\bibfnamefont{G.}~\bibnamefont{Ellis}} \bibnamefont{et~al.},
  in \emph{\bibinfo{booktitle}{Proceedings of the XLVII Enrico Fermi Summer
  School, Enrico Fermi, Varenna, Italy}} (\bibinfo{publisher}{Academic Press,
  New York and London}, \bibinfo{year}{1971}), p. \bibinfo{pages}{104–182}.

\bibitem[{\citenamefont{Ellis}(2009)}]{ellis2009rep}
\bibinfo{author}{\bibfnamefont{G.~F.} \bibnamefont{Ellis}},
  \bibinfo{journal}{General Relativity and Gravitation}
  \textbf{\bibinfo{volume}{41}}, \bibinfo{pages}{581} (\bibinfo{year}{2009}).

\bibitem[{\citenamefont{Schneider et~al.}(1992)\citenamefont{Schneider, Ehlers,
  and Falco}}]{schneider1992}
\bibinfo{author}{\bibfnamefont{P.}~\bibnamefont{Schneider}},
  \bibinfo{author}{\bibfnamefont{J.}~\bibnamefont{Ehlers}}, \bibnamefont{and}
  \bibinfo{author}{\bibfnamefont{E.}~\bibnamefont{Falco}},
  \emph{\bibinfo{title}{Gravitational lenses}} (\bibinfo{publisher}{Springer},
  \bibinfo{year}{1992}).

\bibitem[{\citenamefont{Uzun}(2020)}]{uzun2020}
\bibinfo{author}{\bibfnamefont{N.}~\bibnamefont{Uzun}},
  \bibinfo{journal}{Classical and Quantum Gravity}
  \textbf{\bibinfo{volume}{37}}, \bibinfo{pages}{045002}
  (\bibinfo{year}{2020}).

\bibitem[{\citenamefont{Korzy{\'n}ski and Villa}(2020)}]{korzynski2020}
\bibinfo{author}{\bibfnamefont{M.}~\bibnamefont{Korzy{\'n}ski}}
  \bibnamefont{and} \bibinfo{author}{\bibfnamefont{E.}~\bibnamefont{Villa}},
  \bibinfo{journal}{Physical Review D} \textbf{\bibinfo{volume}{101}},
  \bibinfo{pages}{063506} (\bibinfo{year}{2020}).

\bibitem[{\citenamefont{Serbenta and Korzy{\'n}ski}(2021)}]{serbenta2021}
\bibinfo{author}{\bibfnamefont{J.}~\bibnamefont{Serbenta}} \bibnamefont{and}
  \bibinfo{author}{\bibfnamefont{M.}~\bibnamefont{Korzy{\'n}ski}},
  \bibinfo{journal}{arXiv:2111.02880}  (\bibinfo{year}{2021}).

\bibitem[{\citenamefont{Senovilla}(1998)}]{senovilla1998}
\bibinfo{author}{\bibfnamefont{J.~M.} \bibnamefont{Senovilla}},
  \bibinfo{journal}{General Relativity and Gravitation}
  \textbf{\bibinfo{volume}{30}}, \bibinfo{pages}{701} (\bibinfo{year}{1998}).

\bibitem[{\citenamefont{Curiel}(2017)}]{curiel2017}
\bibinfo{author}{\bibfnamefont{E.}~\bibnamefont{Curiel}}, in
  \emph{\bibinfo{booktitle}{Towards a theory of spacetime theories}}
  (\bibinfo{publisher}{Springer}, \bibinfo{year}{2017}), pp.
  \bibinfo{pages}{43--104}.

\bibitem[{\citenamefont{Maeda and Mart{\'\i}nez}(2020)}]{maeda2020}
\bibinfo{author}{\bibfnamefont{H.}~\bibnamefont{Maeda}} \bibnamefont{and}
  \bibinfo{author}{\bibfnamefont{C.}~\bibnamefont{Mart{\'\i}nez}},
  \bibinfo{journal}{Progress of Theoretical and Experimental Physics}
  \textbf{\bibinfo{volume}{2020}}, \bibinfo{pages}{043E02}
  (\bibinfo{year}{2020}).

\bibitem[{\citenamefont{Catto{\"e}n and Visser}(2008)}]{cattoen2008}
\bibinfo{author}{\bibfnamefont{C.}~\bibnamefont{Catto{\"e}n}} \bibnamefont{and}
  \bibinfo{author}{\bibfnamefont{M.}~\bibnamefont{Visser}},
  \bibinfo{journal}{Classical and Quantum Gravity}
  \textbf{\bibinfo{volume}{25}}, \bibinfo{pages}{165013}
  (\bibinfo{year}{2008}).

\bibitem[{\citenamefont{Barcelo and Visser}(2002)}]{barcelo2002}
\bibinfo{author}{\bibfnamefont{C.}~\bibnamefont{Barcelo}} \bibnamefont{and}
  \bibinfo{author}{\bibfnamefont{M.}~\bibnamefont{Visser}},
  \bibinfo{journal}{International Journal of Modern Physics D}
  \textbf{\bibinfo{volume}{11}}, \bibinfo{pages}{1553} (\bibinfo{year}{2002}).

\bibitem[{\citenamefont{Rubakov}(2014)}]{rubakov2014}
\bibinfo{author}{\bibfnamefont{V.~A.} \bibnamefont{Rubakov}},
  \bibinfo{journal}{Physics-Uspekhi} \textbf{\bibinfo{volume}{57}},
  \bibinfo{pages}{128} (\bibinfo{year}{2014}).

\bibitem[{\citenamefont{Colg{\'a}in and
  Sheikh-Jabbari}(2021)}]{colgain2021critique}
\bibinfo{author}{\bibfnamefont{E.~{\'O}.} \bibnamefont{Colg{\'a}in}}
  \bibnamefont{and}
  \bibinfo{author}{\bibfnamefont{M.}~\bibnamefont{Sheikh-Jabbari}},
  \bibinfo{journal}{Classical and Quantum Gravity}
  \textbf{\bibinfo{volume}{38}}, \bibinfo{pages}{177001}
  (\bibinfo{year}{2021}).

\bibitem[{\citenamefont{Santos et~al.}(2007)\citenamefont{Santos, Alcaniz,
  Reboucas, and Carvalho}}]{santos2007}
\bibinfo{author}{\bibfnamefont{J.}~\bibnamefont{Santos}},
  \bibinfo{author}{\bibfnamefont{J.}~\bibnamefont{Alcaniz}},
  \bibinfo{author}{\bibfnamefont{M.}~\bibnamefont{Reboucas}}, \bibnamefont{and}
  \bibinfo{author}{\bibfnamefont{F.}~\bibnamefont{Carvalho}},
  \bibinfo{journal}{Physical Review D} \textbf{\bibinfo{volume}{76}},
  \bibinfo{pages}{083513} (\bibinfo{year}{2007}).

\bibitem[{\citenamefont{Sotiriou and Faraoni}(2010)}]{sotiriou2010}
\bibinfo{author}{\bibfnamefont{T.~P.} \bibnamefont{Sotiriou}} \bibnamefont{and}
  \bibinfo{author}{\bibfnamefont{V.}~\bibnamefont{Faraoni}},
  \bibinfo{journal}{Reviews of Modern Physics} \textbf{\bibinfo{volume}{82}},
  \bibinfo{pages}{451} (\bibinfo{year}{2010}).

\bibitem[{\citenamefont{Baccetti et~al.}(2012)\citenamefont{Baccetti,
  Martin-Moruno, and Visser}}]{baccetti2012}
\bibinfo{author}{\bibfnamefont{V.}~\bibnamefont{Baccetti}},
  \bibinfo{author}{\bibfnamefont{P.}~\bibnamefont{Martin-Moruno}},
  \bibnamefont{and} \bibinfo{author}{\bibfnamefont{M.}~\bibnamefont{Visser}},
  \bibinfo{journal}{Journal of High Energy Physics}
  \textbf{\bibinfo{volume}{08}}, \bibinfo{pages}{148} (\bibinfo{year}{2012}).

\bibitem[{\citenamefont{Capozziello et~al.}(2014)\citenamefont{Capozziello,
  Lobo, and Mimoso}}]{capozziello2014}
\bibinfo{author}{\bibfnamefont{S.}~\bibnamefont{Capozziello}},
  \bibinfo{author}{\bibfnamefont{F.~S.} \bibnamefont{Lobo}}, \bibnamefont{and}
  \bibinfo{author}{\bibfnamefont{J.~P.} \bibnamefont{Mimoso}},
  \bibinfo{journal}{Physics Letters B} \textbf{\bibinfo{volume}{730}},
  \bibinfo{pages}{280} (\bibinfo{year}{2014}).

\bibitem[{\citenamefont{Capozziello et~al.}(2015)\citenamefont{Capozziello,
  Lobo, and Mimoso}}]{capozziello2015}
\bibinfo{author}{\bibfnamefont{S.}~\bibnamefont{Capozziello}},
  \bibinfo{author}{\bibfnamefont{F.~S.} \bibnamefont{Lobo}}, \bibnamefont{and}
  \bibinfo{author}{\bibfnamefont{J.~P.} \bibnamefont{Mimoso}},
  \bibinfo{journal}{Physical Review D} \textbf{\bibinfo{volume}{91}},
  \bibinfo{pages}{124019} (\bibinfo{year}{2015}).

\bibitem[{\citenamefont{Bolejko et~al.}(2020)\citenamefont{Bolejko, Cinus, and
  Roukema}}]{bolejko2020}
\bibinfo{author}{\bibfnamefont{K.}~\bibnamefont{Bolejko}},
  \bibinfo{author}{\bibfnamefont{M.}~\bibnamefont{Cinus}}, \bibnamefont{and}
  \bibinfo{author}{\bibfnamefont{B.~F.} \bibnamefont{Roukema}},
  \bibinfo{journal}{Physical Review D} \textbf{\bibinfo{volume}{101}},
  \bibinfo{pages}{104046} (\bibinfo{year}{2020}).

\bibitem[{\citenamefont{Speziale}(2018)}]{speziale2018}
\bibinfo{author}{\bibfnamefont{S.}~\bibnamefont{Speziale}},
  \bibinfo{journal}{Physical Review D} \textbf{\bibinfo{volume}{98}},
  \bibinfo{pages}{084029} (\bibinfo{year}{2018}).

\bibitem[{\citenamefont{Santana et~al.}(2017)\citenamefont{Santana, Calv{\~a}o,
  Reis, and Siffert}}]{santana2017}
\bibinfo{author}{\bibfnamefont{L.~T.} \bibnamefont{Santana}},
  \bibinfo{author}{\bibfnamefont{M.~O.} \bibnamefont{Calv{\~a}o}},
  \bibinfo{author}{\bibfnamefont{R.~R.} \bibnamefont{Reis}}, \bibnamefont{and}
  \bibinfo{author}{\bibfnamefont{B.~B.} \bibnamefont{Siffert}},
  \bibinfo{journal}{Physical Review D} \textbf{\bibinfo{volume}{95}},
  \bibinfo{pages}{061501} (\bibinfo{year}{2017}).

\bibitem[{\citenamefont{Korzy\ifmmode~\acute{n}\else \'{n}\fi{}ski
  et~al.}(2021)\citenamefont{Korzy\ifmmode~\acute{n}\else \'{n}\fi{}ski,
  Mi\ifmmode~\acute{s}\else \'{s}\fi{}kiewicz, and Serbenta}}]{korzynski2021}
\bibinfo{author}{\bibfnamefont{M.}~\bibnamefont{Korzy\ifmmode~\acute{n}\else
  \'{n}\fi{}ski}},
  \bibinfo{author}{\bibfnamefont{J.}~\bibnamefont{Mi\ifmmode~\acute{s}\else
  \'{s}\fi{}kiewicz}}, \bibnamefont{and}
  \bibinfo{author}{\bibfnamefont{J.}~\bibnamefont{Serbenta}},
  \bibinfo{journal}{Phys. Rev. D} \textbf{\bibinfo{volume}{104}},
  \bibinfo{pages}{024026} (\bibinfo{year}{2021}).

\bibitem[{\citenamefont{Korzy{\'n}ski and Kopi{\'n}ski}(2018)}]{korzynski2018}
\bibinfo{author}{\bibfnamefont{M.}~\bibnamefont{Korzy{\'n}ski}}
  \bibnamefont{and}
  \bibinfo{author}{\bibfnamefont{J.}~\bibnamefont{Kopi{\'n}ski}},
  \bibinfo{journal}{Journal of Cosmology and Astroparticle Physics}
  \textbf{\bibinfo{volume}{03}}, \bibinfo{pages}{012} (\bibinfo{year}{2018}).

\bibitem[{\citenamefont{Sachs}(1961)}]{sachs1961}
\bibinfo{author}{\bibfnamefont{R.}~\bibnamefont{Sachs}},
  \bibinfo{journal}{Proceedings of the Royal Society of London. Series A.
  Mathematical and Physical Sciences} \textbf{\bibinfo{volume}{264}},
  \bibinfo{pages}{309} (\bibinfo{year}{1961}).

\bibitem[{\citenamefont{Schimd et~al.}(2005)\citenamefont{Schimd, Uzan, and
  Riazuelo}}]{schimd2005}
\bibinfo{author}{\bibfnamefont{C.}~\bibnamefont{Schimd}},
  \bibinfo{author}{\bibfnamefont{J.-P.} \bibnamefont{Uzan}}, \bibnamefont{and}
  \bibinfo{author}{\bibfnamefont{A.}~\bibnamefont{Riazuelo}},
  \bibinfo{journal}{Physical Review D} \textbf{\bibinfo{volume}{71}},
  \bibinfo{pages}{083512} (\bibinfo{year}{2005}).

\bibitem[{\citenamefont{Wald}(2010)}]{wald2010}
\bibinfo{author}{\bibfnamefont{R.~M.} \bibnamefont{Wald}},
  \emph{\bibinfo{title}{General relativity}} (\bibinfo{publisher}{University of
  Chicago press}, \bibinfo{year}{2010}).

\bibitem[{\citenamefont{Poisson}(2004)}]{poisson2004}
\bibinfo{author}{\bibfnamefont{E.}~\bibnamefont{Poisson}},
  \emph{\bibinfo{title}{A relativist's toolkit: the mathematics of black-hole
  mechanics}} (\bibinfo{publisher}{Cambridge university press},
  \bibinfo{year}{2004}).

\bibitem[{\citenamefont{Galloway}(2004)}]{galloway2004}
\bibinfo{author}{\bibfnamefont{G.~J.} \bibnamefont{Galloway}}, in
  \emph{\bibinfo{booktitle}{The Einstein equations and the large scale behavior
  of gravitational fields}} (\bibinfo{publisher}{Springer},
  \bibinfo{year}{2004}), pp. \bibinfo{pages}{379--400}.

\bibitem[{\citenamefont{R{\"a}s{\"a}nen}(2014)}]{rasanen2014}
\bibinfo{author}{\bibfnamefont{S.}~\bibnamefont{R{\"a}s{\"a}nen}},
  \bibinfo{journal}{Journal of Cosmology and Astroparticle Physics}
  \textbf{\bibinfo{volume}{03}}, \bibinfo{pages}{035} (\bibinfo{year}{2014}).

\bibitem[{\citenamefont{Rosquist}(1988)}]{rosquist1988}
\bibinfo{author}{\bibfnamefont{K.}~\bibnamefont{Rosquist}},
  \bibinfo{journal}{The Astrophysical Journal} \textbf{\bibinfo{volume}{331}},
  \bibinfo{pages}{648} (\bibinfo{year}{1988}).

\bibitem[{\citenamefont{Jordan et~al.}(1961)\citenamefont{Jordan, Ehlers, and
  Sachs}}]{jordan1961}
\bibinfo{author}{\bibfnamefont{P.}~\bibnamefont{Jordan}},
  \bibinfo{author}{\bibfnamefont{J.}~\bibnamefont{Ehlers}}, \bibnamefont{and}
  \bibinfo{author}{\bibfnamefont{R.~K.} \bibnamefont{Sachs}},
  \emph{\bibinfo{title}{Beitr{\"a}ge zur Theorie der reinen
  Gravitationsstrahlung: strenge L{\"o}sungen der Feldgleichungen der
  allgemeinen Relativit{\"a}tstheorie II}}, \bibinfo{number}{1-6}
  (\bibinfo{publisher}{Verlag der Akademie der Wissenschaften und der
  Literatur; in Kommission bei F. Steiner, Wiesbaden; Abhandlungen der
  Mathematisch-naturwissenschaftlichen Klasse Nr 1}, \bibinfo{year}{1961}).

\bibitem[{\citenamefont{Jordan et~al.}(2013)\citenamefont{Jordan, Ehlers, and
  Sachs}}]{jordan2013rep}
\bibinfo{author}{\bibfnamefont{P.}~\bibnamefont{Jordan}},
  \bibinfo{author}{\bibfnamefont{J.}~\bibnamefont{Ehlers}}, \bibnamefont{and}
  \bibinfo{author}{\bibfnamefont{R.~K.} \bibnamefont{Sachs}},
  \bibinfo{journal}{General Relativity and Gravitation}
  \textbf{\bibinfo{volume}{45}}, \bibinfo{pages}{2691} (\bibinfo{year}{2013}).

\bibitem[{\citenamefont{Kasai}(1988)}]{kasai1988}
\bibinfo{author}{\bibfnamefont{M.}~\bibnamefont{Kasai}},
  \bibinfo{journal}{Progress of theoretical physics}
  \textbf{\bibinfo{volume}{79}}, \bibinfo{pages}{777} (\bibinfo{year}{1988}).

\bibitem[{\citenamefont{McMillan}(2011)}]{mcmillan2011}
\bibinfo{author}{\bibfnamefont{P.~J.} \bibnamefont{McMillan}},
  \bibinfo{journal}{Monthly Notices of the Royal Astronomical Society}
  \textbf{\bibinfo{volume}{414}}, \bibinfo{pages}{2446} (\bibinfo{year}{2011}).

\bibitem[{\citenamefont{Sanna et~al.}(2017)\citenamefont{Sanna, Reid, Dame,
  Menten, and Brunthaler}}]{sanna2017}
\bibinfo{author}{\bibfnamefont{A.}~\bibnamefont{Sanna}},
  \bibinfo{author}{\bibfnamefont{M.~J.} \bibnamefont{Reid}},
  \bibinfo{author}{\bibfnamefont{T.~M.} \bibnamefont{Dame}},
  \bibinfo{author}{\bibfnamefont{K.~M.} \bibnamefont{Menten}},
  \bibnamefont{and}
  \bibinfo{author}{\bibfnamefont{A.}~\bibnamefont{Brunthaler}},
  \bibinfo{journal}{Science} \textbf{\bibinfo{volume}{358}},
  \bibinfo{pages}{227} (\bibinfo{year}{2017}).

\bibitem[{\citenamefont{Reid and Honma}(2014)}]{reid2014}
\bibinfo{author}{\bibfnamefont{M.~J.} \bibnamefont{Reid}} \bibnamefont{and}
  \bibinfo{author}{\bibfnamefont{M.}~\bibnamefont{Honma}},
  \bibinfo{journal}{Annual Review of Astronomy and Astrophysics}
  \textbf{\bibinfo{volume}{52}}, \bibinfo{pages}{339} (\bibinfo{year}{2014}).

\bibitem[{\citenamefont{Paragi et~al.}(2014)\citenamefont{Paragi, Godfrey,
  Reynolds, Rioja, Deller, Zhang, Gurvits, Bietenholz, Szomoru, Bignall
  et~al.}}]{paragi2014}
\bibinfo{author}{\bibfnamefont{Z.}~\bibnamefont{Paragi}},
  \bibinfo{author}{\bibfnamefont{L.}~\bibnamefont{Godfrey}},
  \bibinfo{author}{\bibfnamefont{C.}~\bibnamefont{Reynolds}},
  \bibinfo{author}{\bibfnamefont{M.}~\bibnamefont{Rioja}},
  \bibinfo{author}{\bibfnamefont{A.}~\bibnamefont{Deller}},
  \bibinfo{author}{\bibfnamefont{B.}~\bibnamefont{Zhang}},
  \bibinfo{author}{\bibfnamefont{L.}~\bibnamefont{Gurvits}},
  \bibinfo{author}{\bibfnamefont{M.}~\bibnamefont{Bietenholz}},
  \bibinfo{author}{\bibfnamefont{A.}~\bibnamefont{Szomoru}},
  \bibinfo{author}{\bibfnamefont{H.}~\bibnamefont{Bignall}},
  \bibnamefont{et~al.}, \bibinfo{journal}{arXiv:1412.5971}
  (\bibinfo{year}{2014}).

\bibitem[{\citenamefont{Rioja and Dodson}(2020)}]{rioja2020}
\bibinfo{author}{\bibfnamefont{M.~J.} \bibnamefont{Rioja}} \bibnamefont{and}
  \bibinfo{author}{\bibfnamefont{R.}~\bibnamefont{Dodson}},
  \bibinfo{journal}{The Astronomy and Astrophysics Review}
  \textbf{\bibinfo{volume}{28}}, \bibinfo{pages}{6} (\bibinfo{year}{2020}).

\bibitem[{\citenamefont{Kardashev}(1986)}]{kardashev}
\bibinfo{author}{\bibfnamefont{N.~S.} \bibnamefont{Kardashev}},
  \bibinfo{journal}{Soviet astronomy} \textbf{\bibinfo{volume}{30}},
  \bibinfo{pages}{501} (\bibinfo{year}{1986}).

\bibitem[{\citenamefont{Marcori et~al.}(2018)\citenamefont{Marcori, Pitrou,
  Uzan, and Pereira}}]{Marcori:2018cw}
\bibinfo{author}{\bibfnamefont{O.~H.} \bibnamefont{Marcori}},
  \bibinfo{author}{\bibfnamefont{C.}~\bibnamefont{Pitrou}},
  \bibinfo{author}{\bibfnamefont{J.-P.} \bibnamefont{Uzan}}, \bibnamefont{and}
  \bibinfo{author}{\bibfnamefont{T.~S.} \bibnamefont{Pereira}},
  \bibinfo{journal}{Physical Review D} \textbf{\bibinfo{volume}{98}},
  \bibinfo{pages}{023517} (\bibinfo{year}{2018}).

\bibitem[{\citenamefont{McCrea}(1935)}]{mccrea1935}
\bibinfo{author}{\bibfnamefont{W.}~\bibnamefont{McCrea}},
  \bibinfo{journal}{Zeitschrift fur Astrophysik} \textbf{\bibinfo{volume}{9}},
  \bibinfo{pages}{290} (\bibinfo{year}{1935}).

\bibitem[{\citenamefont{Kardashev et~al.}(1973)\citenamefont{Kardashev,
  Parijskij, and Umarbaeva}}]{kardashev1973}
\bibinfo{author}{\bibfnamefont{N.}~\bibnamefont{Kardashev}},
  \bibinfo{author}{\bibfnamefont{Y.~N.} \bibnamefont{Parijskij}},
  \bibnamefont{and}
  \bibinfo{author}{\bibfnamefont{N.}~\bibnamefont{Umarbaeva}},
  \bibinfo{journal}{Astrofizicheskie Issledovaniia Izvestiya Spetsial'noj
  Astrofizicheskoj Observatorii} \textbf{\bibinfo{volume}{5}},
  \bibinfo{pages}{16} (\bibinfo{year}{1973}).

\bibitem[{\citenamefont{Novikov}(1977)}]{novikov1977}
\bibinfo{author}{\bibfnamefont{I.}~\bibnamefont{Novikov}},
  \bibinfo{journal}{Astronomicheskii Zhurnal} \textbf{\bibinfo{volume}{54}},
  \bibinfo{pages}{722 [Soviet Astronomy \textbf{21}, 406]}
  (\bibinfo{year}{1977}).

\bibitem[{\citenamefont{Novikov and Popova}(1978)}]{novikov1978}
\bibinfo{author}{\bibfnamefont{I.}~\bibnamefont{Novikov}} \bibnamefont{and}
  \bibinfo{author}{\bibfnamefont{A.}~\bibnamefont{Popova}},
  \bibinfo{journal}{Soviet Astronomy} \textbf{\bibinfo{volume}{22}},
  \bibinfo{pages}{649} (\bibinfo{year}{1978}).

\bibitem[{\citenamefont{Quercellini et~al.}(2012)\citenamefont{Quercellini,
  Amendola, Balbi, Cabella, and Quartin}}]{Quercellini:2010zr}
\bibinfo{author}{\bibfnamefont{C.}~\bibnamefont{Quercellini}},
  \bibinfo{author}{\bibfnamefont{L.}~\bibnamefont{Amendola}},
  \bibinfo{author}{\bibfnamefont{A.}~\bibnamefont{Balbi}},
  \bibinfo{author}{\bibfnamefont{P.}~\bibnamefont{Cabella}}, \bibnamefont{and}
  \bibinfo{author}{\bibfnamefont{M.}~\bibnamefont{Quartin}},
  \bibinfo{journal}{Physics Reports} \textbf{\bibinfo{volume}{521}},
  \bibinfo{pages}{95} (\bibinfo{year}{2012}).

\end{thebibliography}

\end{document}